\documentclass[journal]{IEEEtran}

\usepackage[T1]{fontenc}
\usepackage[latin9]{inputenc}
\usepackage{float}
\usepackage{amsmath}
\usepackage{amsthm}
\usepackage{amssymb}
\usepackage{graphicx}
\usepackage{stfloats}
\usepackage{cite} 
\usepackage{xcolor}
\usepackage{bbm}
\usepackage{multirow}
\usepackage{multicol}
\usepackage{caption}
\usepackage{subfig}
\usepackage{makecell}

\makeatletter

\floatstyle{ruled}
\newfloat{algorithm}{tbp}{loa}
\providecommand{\algorithmname}{Algorithm}
\floatname{algorithm}{\protect\algorithmname}

\theoremstyle{plain}
\newtheorem{thm}{\protect\theoremname}
\theoremstyle{plain}
\newtheorem{prop}[thm]{\protect\propositionname}

\IEEEoverridecommandlockouts
\usepackage{cite}
\usepackage{amsfonts}\usepackage{algorithmic}
\usepackage{textcomp}
\usepackage{etoolbox}

\renewcommand{\Re}{\operatorname{Re}}

\def\BibTeX{{\rm B\kern-.05em{\sc i\kern-.025em b}\kern-.08em
    T\kern-.1667em\lower.7ex\hbox{E}\kern-.125emX}}

\patchcmd{\@maketitle}
  {\addvspace{0.5\baselineskip}\egroup}
  {\addvspace{-1\baselineskip}\egroup}
  {}
  {}

\makeatother

\providecommand{\propositionname}{Proposition}
\providecommand{\theoremname}{Theorem}

\begin{document}

\title{Efficient Design of Fronthaul-Constrained Uplink Reception for Cell-Free XL-MIMO}

\author{Dogon Kim, \textit{Graduate Student Member}, \textit{IEEE}, Hyunmin Noh, \textit{Member}, \textit{IEEE}, \\and Seok-Hwan Park, \textit{Senior Member}, \textit{IEEE} \thanks{

This work was supported in part by the Regional Innovation System and Education (RISE) Program through Jeonbuk RISE Center, funded by the Ministry of Education (MOE) and Jeonbuk State, Republic of Korea, under Grant 2025-RISE-13-JBU; in part by the National Research Foundation (NRF) funded by the Ministry of Science and ICT (MSIT) under Grant RS-2026-25468472; and in part by the Institute of Information $\&$ Communications Technology Planning $\&$ Evaluation(IITP) funded by the MSIT under Grant IITP-2026-RS-2024-00439292. \textit{(Corresponding authors: Hyunmin Noh; Seok-Hwan Park.)}

D. Kim and H. Noh are with the Division of Electronic Engineering, Jeonbuk
National University, Jeonju, Korea (email: kdg8728@jbnu.ac.kr, hmnoh@jbnu.ac.kr). S.-H. Park is with the School of Electrical Engineering, Hanyang University, Ansan, South Korea (seokhwanpark@hanyang.ac.kr).}
}
\maketitle
\begin{abstract}
With the evolution of multiple-input multiple-output (MIMO) technology toward extremely large (XL) MIMO systems comprising hundreds of, or more, antennas, this work investigates scalable and fronthaul-efficient reception design for the uplink of cell-free (CF) XL-MIMO systems. In such systems, the uplink signals transmitted by mobile user equipments (UEs) are jointly decoded at a central processing unit (CPU) connected to distributed access points (APs) via finite-capacity fronthaul links. We address the joint optimization of linear transform matrices, used by the APs to reduce the signal dimension and fronthaul load, and fronthaul compression strategies to maximize the uplink sum-rate. A fractional programming (FP)-based iterative algorithm is first developed, followed by a reduced-complexity variant, termed accelerated FP (A-FP), along with its decentralized implementation whose fronthaul overhead remains independent of the number of AP antennas.
Numerical results show that the proposed A-FP scheme significantly reduces computational complexity compared to FP implemented with general-purpose solvers, while substantially outperforming scalable baseline schemes that rely solely on local channel state information.
\end{abstract}


\theoremstyle{theorem}
\newtheorem{theorem}{Theorem} 
\theoremstyle{proposition}
\newtheorem{proposition}{Proposition} 
\theoremstyle{lemma}
\newtheorem{lemma}{Lemma} 
\theoremstyle{corollary}
\newtheorem{corollary}{Corollary} 
\theoremstyle{definition}
\newtheorem{definition}{Definition}
\theoremstyle{remark}
\newtheorem{remark}{Remark}

\section{Introduction}

To enable immersive and massive communications in the 6G era, multiple-input multiple-output (MIMO) techniques are expected to evolve by increasing the number of antennas by at least an order of magnitude, resulting in systems with several hundreds or even thousands of antennas \cite{Xu:JSTSP25}.
This trend necessitates the development of scalable and efficient signal processing design for extremely large (XL) MIMO systems.
In this work, we investigate this challenge in the context of the uplink of cell-free (CF) XL-MIMO systems, where the design problem becomes particularly challenging for two key reasons. First, as a central processing unit (CPU) coordinates distributed access points (APs) via fronthaul links, their signal processing strategies need to be jointly optimized, leading to a high-dimensional and computationally intensive optimization problem.
Second, since these strategies typically affect the system objective, e.g., sum-rate, in a coupled manner, global channel state information (CSI) is required for joint optimization. Hence, it is crucial to develop a fronthaul-efficient decentralized optimization algorithm, that minimizes the fronthaul overhead exchanged between the CPU and APs during the optimization process.

Reference \cite{Zhou:TSP16} studied the optimization of fronthaul compression for the uplink of CF XL-MIMO systems, but the developed algorithm relies on an iterative process, which does not scale well with the number of AP antennas. 
To address this limitation, References \cite{Liu:Cambridge17, Wiffen:WCNC20} proposed applying a linear transformation before compression to reduce both the signal dimension and fronthaul overhead.
However, for practical reasons, these transformation matrices were designed based solely on local CSI, resulting in significantly suboptimal performance.
Scalable algorithms for optimizing linear transformation matrices were later explored in \cite{Xu:TSP24, Zhuang:JSTSP25}, but they considered compression only as a means of dimensionality reduction without explicitly accounting for the effects of finite-capacity fronthaul links.

Motivated by the above discussion, we address the joint optimization of linear transformation matrices and fronthaul compression strategies for the uplink of CF XL-MIMO systems. 
Specifically, we aim to develop a fronthaul-efficient and scalable optimization algorithm.
To this end, we first derive an iterative algorithm based on fractional programming (FP) \cite{Shen:TN19} and Fenchel's inequality \cite{Zhou:TSP16} to solve the sum-rate maximization problem.
Since the resulting FP algorithm\footnote{The proposed FP-based algorithm can be seen as an instance of the successive convex approximation (SCA) framework, in the sense that it iteratively solves a sequence of convex problems obtained by locally approximating non-convex functions with tight surrogate bounds, thereby achieving monotonically non-decreasing objective values over the iterations.} requires both high computational complexity and global CSI, inspired by the low-complexity FP approach proposed in \cite{Kim:WCL25} for integrated sensing and communication (ISAC) beamforming design, we develop a reduced-complexity variant referred to as accelerated FP (A-FP) scheme. In the proposed A-FP scheme, the overall transformation and fronthaul qauntization noise covariance matrices are partitioned into per-AP block variables, which are sequentially updated in closed form.
We then analyze the computational complexity of both the FP and A-FP schemes and present a decentralized implementation of the proposed A-FP scheme, whose fronthaul overhead in both uplink and downlink directions remains independent of the number of AP antennas.
Through numerical results, we demonstrate that the A-FP scheme achieves nearly identical performance to FP implemented with general-purpose solvers while dramatically reducing computational complexity, and significantly outperforms baseline scalable schemes based solely on local CSI.

The main contributions of this work are summarized as follows:
\begin{itemize}
    \item We formulate a sum-rate maximization problem that jointly optimizes the linear transformation and fronthaul compression strategies for uplink CF XL-MIMO systems under finite-capacity fronthaul constraints, and develop an FP-based algorithm that finds a stationary point but incurs high computational complexity and global CSI requirements.
    \item We propose a reduced-complexity variant, termed A-FP, which enables per-AP sequential updates in closed form and supports decentralized implementation with fronthaul overhead independent of the number of AP antennas.
    \item Numerical results show that A-FP achieves near-optimal performance compared to FP implemented with general-purpose solvers while significantly reducing computational complexity and outperforming scalable baselines relying solely on local CSI.
\end{itemize}

\section{System Model and Problem Definition} \label{sec:System-Model}

We consider the uplink of a CF XL-MIMO system, where $K$ single-antenna user equipments (UEs) send independent messages to a CPU through $L$ APs, each equipped with $M$ antennas.
Our focus is on the XL-MIMO regime, 
characterized by each AP having a very large number of antennas, i.e., $M \gg 1$.
We assume that the fronthaul link connecting each AP $i$ to the CPU has a limited capacity of $C_F$ bps/Hz. The synchronization among APs and between the APs and UEs is assumed to be perfect.
We define index sets $\mathcal{K} = \{1,2,\ldots,K\}$ and $\mathcal{L} = \{1,2,\ldots,L\}$.

\subsection{Channel Model and Signal Processing Models} \label{sub:Uplink-Channel-Model}

The received signal vector $\mathbf{y}_i\in\mathbb{C}^{M}$ at AP $i$ is given by
\begin{align}
    \mathbf{y}_i = \mathbf{H}_i \mathbf{x} + \mathbf{z}_i, \label{UL-channel-model}
\end{align}
where $\mathbf{x} = [x_1 \, \cdots \, x_K]^T\in\mathbb{C}^K$ denotes the transmitted signal vector with each $x_k\sim\mathcal{CN}(0,p_k)$ representing the transmitted signal by UE $k$, $\mathbf{H}_{i} = [\mathbf{h}_{i,1} \, \cdots \, \mathbf{h}_{i,K}]\in\mathbb{C}^{M\times K}$ represents the channel matrix from the UEs to AP $i$ with $\mathbf{h}_{i,k}\in\mathbb{C}^M$ denoting the channel vector from UE $k$, and $\mathbf{z}_i\sim\mathcal{CN}(\mathbf{0}, \sigma_z^2\mathbf{I})$ is the additive noise vector.
Throughout the paper, we assume that each AP $i$ has knowledge of its \textit{local CSI} $\{\mathbf{h}_{i,k}\}_{k\in\mathcal{K}}$ which collects the channel vectors between itself and all UEs.




To convey the high-dimensional received signal $\mathbf{y}_i\in\mathbb{C}^{M\times 1}$ ($M\gg 1$) to the CPU over a finite-capacity fronthaul link, AP $i$ performs two sequential operations: \textit{linear transformation} and \textit{fronthaul compression}, as described below.

AP $i$ first applies a linear transformation to reduce the dimensionality of $\mathbf{y}_i$, generating an $N$-dimensional signal $\mathbf{v}_i\in\mathbb{C}^N$ with $N\ll M$, as
\begin{align}
    \mathbf{v}_i = \mathbf{W}_i \mathbf{y}_i, \label{eq:linear-transform}
\end{align}
where $\mathbf{W}_i\in\mathbb{C}^{N\times M}$ denotes the linear transformation matrix.

The transformed signal $\mathbf{v}_i\in\mathbb{C}^N$ is then quantized and compressed for transmission over the finite-capacity fronthaul link.
Following the Gaussian test channel model in \cite{Zhou:TSP16}, the quantized signal $\hat{\mathbf{v}}_i$ recovered at the CPU can be expressed as
\begin{align}
    \hat{\mathbf{v}}_i = \mathbf{v}_i + \mathbf{e}_i, \label{eq:fronthaul-compression}
\end{align}
with the quantization noise $\mathbf{e}_i$ being uncorrelated with $\mathbf{v}_i$ and distributed as $\mathbf{e}\sim\mathcal{CN}(\mathbf{0}, \boldsymbol{\Omega}_i)$, where $\boldsymbol{\Omega}_i\succeq \mathbf{0}$ denotes the quantization noise covariance matrix.

To ensure reliable recovery of $\hat{\mathbf{v}}_i$ at the CPU, the following condition should be met \cite{Zhou:TSP16}:
\begin{align}
    \log_2\det\left( 
\mathbf{W}_i \mathbf{S}_i \mathbf{W}_i^H + \boldsymbol{\Omega}_i \right) - \log_2\det\left(\boldsymbol{\Omega}_i\right) \leq C_F, \label{eq:fronthaul-capacity-constraint}
\end{align}
where $\mathbf{S}_i = \mathbf{H}_i \bar{\mathbf{P}} \mathbf{H}_i^H + \sigma_z^2 \mathbf{I}$ with $\bar{\mathbf{P}} = \text{diag}(\{p_k\}_{k\in\mathcal{K}})$.

\subsection{Achievable Rates and Problem Definition}

The CPU decodes each UE's signal $x_k$ from the stacked quantized signal vector $\hat{\mathbf{v}} = [\hat{\mathbf{v}}_1^H \, \cdots \, \hat{\mathbf{v}}_L^H]^H \in \mathbb{C}^{LN}$ given as
\begin{align}
    \hat{\mathbf{v}} = \bar{\mathbf{W}} \sum\nolimits_{k\in\mathcal{K}} \mathbf{h}_k x_k + \bar{\mathbf{W}}\bar{\mathbf{z}} + \bar{\mathbf{e}}. \label{eq:stacked-quantized-signals}
\end{align}
Here we have defined $\mathbf{h}_k = [\mathbf{h}_{1,k}^H \, \cdots \, \mathbf{h}_{L,k}^H]^H \in \mathbb{C}^{L M}$, $\bar{\mathbf{z}} = [\mathbf{z}_1^H \, \cdots \, \mathbf{z}_L^H]^H \in \mathbb{C}^{LM}$, and $\bar{\mathbf{e}} = [\mathbf{e}_1^H \, \cdots \, \mathbf{e}_L^H]^H \in \mathbb{C}^{L N}$.
Accordingly, the data rate $R_k$ for UE $k$ is given by
\begin{align}
    R_k = \log_2\left( 1 + p_k \mathbf{h}_k^H \bar{\mathbf{W}}^H \mathbf{C}_{\text{IF},k}^{-1} \bar{\mathbf{W}} \mathbf{h}_k \right), \label{eq:data-rate}
\end{align}
where $\mathbf{C}_{\text{IF},k} = \sum\nolimits_{k^{\prime}\in\mathcal{K}\setminus\{k\}} p_{k^{\prime}} \bar{\mathbf{W}} \mathbf{h}_{k^{\prime}} \mathbf{h}_{k^{\prime}}^H \bar{\mathbf{W}}^H + \sigma_z^2 \bar{\mathbf{W}}\bar{\mathbf{W}}^H + \bar{\boldsymbol{\Omega}}$ denotes the covariance matrix of the interference-plus-noise term with $\bar{\boldsymbol{\Omega}} = \text{blkdiag}(\{ \boldsymbol{\Omega}_i\}_{i\in\mathcal{L}})$.

Throughout the paper, our objective is to maximize the sum-rate of the UEs $\sum_{k\in\mathcal{K}} R_k$ by jointly optimizing the linear transformation matrices $\mathbf{W} = \{\mathbf{W}_i\}_{i\in\mathcal{L}}$ and the quantization noise covariance matrices $\boldsymbol{\Omega} = \{\boldsymbol{\Omega}_i\}_{i\in\mathcal{L}}$.
Specifically, we aim to solve the following optimization problem:
\begin{subequations} \label{eq:problem-original}
\begin{align}
    \!\!\!\!\underset{ \mathbf{W}, \boldsymbol{\Omega} } {\mathrm{max.}}\,\,\, & \sum\nolimits_{k\in\mathcal{K}} \log_2\left( 1 + p_k \mathbf{h}_k^H \bar{\mathbf{W}}^H \mathbf{C}_{\text{IF},k}^{-1} \bar{\mathbf{W}} \mathbf{h}_k \right) \, \label{eq:problem-original-objective} \\
 \mathrm{s.t. }\quad & \text{(\ref{eq:fronthaul-capacity-constraint}) and } \boldsymbol{\Omega}_i \succeq \mathbf{0}, \,\, \forall i\in\mathcal{L}. \label{eq:problem-original-psd}
\end{align}
\end{subequations}
Solving problem (\ref{eq:problem-original}) presents two main challenges. It first requires global CSI, which is to difficult to obtain in practice. Moreover, even with full knowledge of global CSI, the problem remains difficult due to its inherent non-convexity.

\section{FP-Based Optimization Algorithms} \label{sec:FP-based-algorithms}

In this section, we develop FP-based alternating optimization (AO) algorithms to solve problem (\ref{eq:problem-original}).
We first present a direct application of FP approach to (\ref{eq:problem-original}) and observe that it is challenging to implement the developed algorithm due to high computational complexity and large fronthaul overhead for exchanging CSI.
To overcome these limitations, we propose an accelerated version of the FP-based scheme and discuss its fronthaul-efficient decentralized implementation.

\subsection{FP-Based Optimization} \label{sub:FP-based}

Since directly solving the non-convex problem (\ref{eq:problem-original}) is challenging due to the non-convexity, we first reformulate it into an equivalent problem, as presented in the following proposition.

\begin{prop}[Equivalent Reformulation]
Problem (\ref{eq:problem-original}) can be equivalently reformulated as the following problem:
\begin{subequations} \label{eq:problem-FP}
\begin{align}
    \!\!\!\!\underset{ ^{\mathbf{W}, \boldsymbol{\Omega},}_{ \boldsymbol{\gamma}, \boldsymbol{\theta}, \boldsymbol{\Sigma}} } {\mathrm{max.}}\,\,\, & \sum\nolimits_{k\in\mathcal{K}} 
    \left( \log_2\left(1 + \gamma_k\right) - \frac{\gamma_k}{\ln 2} \right) \label{eq:problem-FP-objective} \\
    & + \sum\nolimits_{k\in\mathcal{K}}\frac{1+\gamma_k}{\ln 2} \Big( 2 \Re\big\{ (\bar{\mathbf{W}}\mathbf{h}_k)^H p_k^{1/2}\boldsymbol{\theta}_k \big\} \nonumber \\
    & - \sum\nolimits_{k^{\prime}\in\mathcal{K}} p_{k^{\prime}} \big|\boldsymbol{\theta}_k^H \bar{\mathbf{W}} \mathbf{h}_{k^{\prime}}\big|^2 - \sigma_z^2 \|\bar{\mathbf{W}}^H\boldsymbol{\theta}_k\|^2 - \boldsymbol{\theta}_k^H \bar{\boldsymbol{\Omega}} \boldsymbol{\theta}_k\Big) \nonumber \\
 \mathrm{s.t. }\,\,\, & \ln\det\left(\boldsymbol{\Sigma}_i\right) + \mathrm{tr}\left( \boldsymbol{\Sigma}_i^{-1} \left( 
\mathbf{W}_i \mathbf{S}_i \mathbf{W}_i^H + \boldsymbol{\Omega}_i \right) \right) - N 
\nonumber \\
& \qquad - \ln\det\left(\boldsymbol{\Omega}_i\right) \leq C_F \ln 2, \, \forall i\in\mathcal{L}, \label{eq:problem-FP-FH} \\
& \boldsymbol{\Omega}_i \succeq \mathbf{0}, \boldsymbol{\Sigma}_i\succ \mathbf{0}, \,\, \forall i\in\mathcal{L}, \,\, \gamma_k \geq 0, \,\, \forall k\in\mathcal{K}, \label{eq:problem-FP-psd}
\end{align}
\end{subequations}
where $\boldsymbol{\gamma} = \{\gamma_k\}_{k\in\mathcal{K}}$, $\boldsymbol{\theta} = \{\boldsymbol{\theta}_k\}_{k\in\mathcal{K}}$, and $\boldsymbol{\Sigma} = \{\boldsymbol{\Sigma}_i\}_{i\in\mathcal{L}}$.

For fixed $\mathbf{W}$ and $\boldsymbol{\Omega}$, the optimal $\{\boldsymbol{\gamma}, \boldsymbol{\theta}\}$, that maximize the objective function in (\ref{eq:problem-FP-objective}), and the optimal $\boldsymbol{\Sigma}_i$, that minimizes the left-hand side of (\ref{eq:problem-FP-FH}), are given by
\begin{subequations} \label{eq:optimal-gamma-theta-Sigma}
\begin{align}
    &\gamma_k = p_k \left(\bar{\mathbf{W}}\mathbf{h}_k\right)^H \mathbf{C}_{\mathrm{IF},k}^{-1} \bar{\mathbf{W}}\mathbf{h}_k, \, k\in\mathcal{K},\label{eq:optimal-gamma} \\
    &\boldsymbol{\theta}_k = p_k^{1/2}\!\left( 
p_k\bar{\mathbf{W}}\mathbf{h}_k \mathbf{h}_k^H\bar{\mathbf{W}}^H \!+\! \mathbf{C}_{\mathrm{IF},k}\right)^{-1} \!\bar{\mathbf{W}}\mathbf{h}_k, \, k\in\mathcal{K},\label{eq:optimal-theta} \\
    &\boldsymbol{\Sigma}_i = \mathbf{W}_i \mathbf{S}_i \mathbf{W}_i^H + \boldsymbol{\Omega}_i, \, i\in\mathcal{L}. \label{eq:optimal-Sigma}
\end{align}
\end{subequations}

\end{prop}
\begin{proof}
    The objective function in (\ref{eq:problem-FP-objective}) is derived by applying the quadratic transform used in the FP framework \cite[Thm. 1]{Shen:TN19} to the original objective function in (\ref{eq:problem-original-objective}). Similarly, the constraint in (\ref{eq:problem-FP-FH}) is obtained by applying Fenchel's inequality for $\log_2\det(\cdot)$ in \cite[Lem. 1]{Zhou:TSP16} to (\ref{eq:fronthaul-capacity-constraint}).
\end{proof}

Compared to the original problem (\ref{eq:problem-original}), the reformulated problem (\ref{eq:problem-FP}) introduces additional optimization variables $\{\boldsymbol{\gamma}, \boldsymbol{\theta}, \boldsymbol{\Sigma}\}$.
However, for these variables kept fixed, the resulting subproblem of optimizing only $\{\mathbf{W}, \boldsymbol{\Omega}\}$ becomes convex and can be efficiently solved using standard convex optimization solvers (e.g., CVX \cite{Grant:CVX20}).
Conversely, for fixed $\{\mathbf{W}, \boldsymbol{\Omega}\}$, the optimal $\{\boldsymbol{\gamma}, \boldsymbol{\theta}, \boldsymbol{\Sigma}\}$ can be found in closed form as (\ref{eq:optimal-gamma-theta-Sigma}).
This structure naturally motivates an AO approach that iteratively updates $\{\boldsymbol{\gamma}, \boldsymbol{\theta}, \boldsymbol{\Sigma}\}$ and $\{\mathbf{W}, \boldsymbol{\Omega}\}$.
It is noted that each iteration of the AO algorithm guarantees a non-decreasing sum-rate value, and hence the algorithm converges to a stationary point of the original problem (\ref{eq:problem-original}).








\subsection{Accelerated FP-Based Optimization (A-FP)} \label{sub:distributed}

It is challenging to implement the FP-based AO scheme developed in Sec. \ref{sub:FP-based} for the following reasons. First, it assumes a centralized scenario, where each AP reports its local CSI to the CPU. The CPU then jointly optimizes $\mathbf{W}$ and $\boldsymbol{\Omega}$ based on the global CSI and sends the optimized variables back to the APs. Consequently, the amount of information exchanged over the fronthaul links is proportional to the number of AP antennas $M$, which becomes prohibitively large in CF XL-MIMO systems with $M\gg 1$.
Second, the computational complexity at the CPU is substantial due to the convex optimization process at each iteration.
To overcome these limitations, we propose in this subsection an accelerated version of the FP-based scheme, termed A-FP, and discuss its complexity and decentralized implementation, in which the amount of data exchanged over the fronthaul links is independent of $M$, in the following subsections.


The computational bottleneck of the FP-based algorithm comes from the process of optimizing the primary variables $\{\mathbf{W}, \boldsymbol{\Omega}\}$ for fixed $\{\boldsymbol{\gamma},\boldsymbol{\theta},\boldsymbol{\Sigma}\}$, which relies on convex solvers.
To resolve this bottleneck, we consider per-AP individual optimization, i.e., optimizing each pair $\{\mathbf{W}_i, \boldsymbol{\Omega}_i\}$ of AP $i$ for fixed other APs' variables $\{\mathbf{W}_j, \boldsymbol{\Omega}_j\}_{j\in\mathcal{L}\setminus\{i\}}$. The goal of this individual optimization on a per-AP basis is to update each $\{\mathbf{W}_i, \boldsymbol{\Omega}_i\}$ with closed-form computations, accelerating the overall optimization process \cite{Kim:WCL25}.

Let us define $\mathbf{w}_i = \text{vec}(\mathbf{W}_i)\in\mathbb{C}^{MN}$.
By removing the terms or constraints that are not dependent on $\{\mathbf{w}_i, \boldsymbol{\Omega}_i\}$, the subproblem for $\{\mathbf{w}_i, \boldsymbol{\Omega}_i\}$ with fixed $\{\mathbf{w}_j, \boldsymbol{\Omega}_j\}_{j\in\mathcal{L}\setminus\{i\}}$ can be stated as
\begingroup
\allowdisplaybreaks
\begin{subequations} \label{eq:problem-perAP-quadratic}
\begin{align}
    \!\!\!\!\underset{ \mathbf{w}_i, \boldsymbol{\Omega}_i} {\mathrm{min.}}\,\,\, &     \mathbf{w}_i^H\mathbf{D}_i\mathbf{w}_i - \Re\left\{\tilde{\mathbf{g}}_i^H\mathbf{w}_i\right\} +\text{tr}\left(\mathbf{T}_i\boldsymbol{\Omega}_i\right)\label{eq:problem-perAP-quadratic-objective} \\
 \mathrm{s.t. }\,\,\, & \|\mathbf{A}_i^H\mathbf{w}_i\|^2 + \text{tr}\left(\boldsymbol{\Sigma}_i^{-1}\boldsymbol{\Omega}_i\right) - \ln \det\left(\boldsymbol{\Omega}_i\right) \leq \tilde{C}_{F,i},
\label{eq:problem-perAP-quadratic-FH} \\
& \boldsymbol{\Omega}_i \succeq \mathbf{0}, \label{eq:problem-perAP-quadratic-psd}
\end{align}
\end{subequations}
\endgroup
where $\mathbf{D}_i \in \mathbb{C}^{MN\times MN}$, $\tilde{\mathbf{g}}_i\in\mathbb{C}^{MN}$, $\mathbf{T}_i\in\mathbb{C}^{N\times N}$, $\mathbf{A}_i\in\mathbb{C}^{MN \times MN}$, and $\tilde{C}_{F,i}$ are defined as
\begingroup
\allowdisplaybreaks
\begin{subequations} \label{eq:def-Di-gtildei-Ti-Ai}
\begin{align}
    & \mathbf{D}_i = \sum\nolimits_{k\in\mathcal{K}} \left(1+\gamma_k\right) \Big( \sum\nolimits_{k^{\prime}\in\mathcal{K}}p_{k^{\prime}} \tilde{\mathbf{h}}_{i,k^{\prime},k}\tilde{\mathbf{h}}_{i,k^{\prime},k}^H \label{eq:def-Di} \\
    & \qquad  + \sigma_z^2\left( \mathbf{I}_M\otimes(\boldsymbol{\theta}_{i,k}\boldsymbol{\theta}_{i,k}^H) \right) \Big), \nonumber \\
    & \tilde{\mathbf{g}}_i = \sum\nolimits_{k\in\mathcal{K}} 2\left(1+\gamma_k\right) \tilde{\mathbf{g}}_{i,k}, \label{eq:def-gtildei} \\
    & \mathbf{T}_i = \sum\nolimits_{k\in\mathcal{K}} \left(1+\gamma_k\right) \boldsymbol{\theta}_{i,k} \boldsymbol{\theta}_{i,k}^H,\label{eq:def-Ti} \\
    & \mathbf{A}_i = (\mathbf{S}_i^{1/2})^* \otimes \boldsymbol{\Sigma}_i^{-1/2}, \label{eq:def-Ai} \\
    & \tilde{C}_{F,i} = C_F \ln 2 - \ln \det \left(\boldsymbol{\Sigma}_i\right) +N, \label{eq:def-CtildeFi}
\end{align}
\end{subequations}
\endgroup
with $\tilde{\mathbf{h}}_{i,k^{\prime},k} = \mathbf{h}_{i,k^{\prime}}^* \otimes \boldsymbol{\theta}_{i,k}$, $\tilde{\mathbf{g}}_{i,k} = \mathbf{g}_{i,k}^* \otimes \boldsymbol{\theta}_{i,k}$,  $\mathbf{g}_{i,k} = p_k^{1/2} \mathbf{h}_{i,k} - \sum_{k^{\prime}\in\mathcal{K}} p_{k^{\prime}} \alpha_{i,k,k^{\prime}}^* \mathbf{h}_{i,k^{\prime}}$, and $\alpha_{i,k,k^{\prime}} = \sum_{i^{\prime}\in\mathcal{L}\setminus\{i\}} \boldsymbol{\theta}_{i^{\prime},k}^H \mathbf{W}_{i^{\prime}} \mathbf{h}_{i^{\prime}, k^{\prime}}$.
Here we have defined $\boldsymbol{\theta}_{i,k}\in\mathbb{C}^{N}$ as the $i$th subvector of $\boldsymbol{\theta}_k$, i.e., $\boldsymbol{\theta}_k = [\boldsymbol{\theta}_{1,k}^H \cdots \boldsymbol{\theta}_{L,k}^H]^H$.

To tackle the problem (\ref{eq:problem-perAP-quadratic}), we consider its Lagrangian dual problem given by
\begin{align}
    \underset{ \mu_i \geq 0} {\mathrm{max}} \,\, \underset{ \mathbf{w}_i, \boldsymbol{\Omega}_i\succeq \mathbf{0}} {\mathrm{min}} \,\, \mathcal{L}\left(\mathbf{w}_i, \boldsymbol{\Omega}_i, \mu_i\right), \label{eq:dual-problem}
\end{align}
where the Lagrangian $\mathcal{L}(\mathbf{w}_i, \boldsymbol{\Omega}_i, \mu_i)$ of (\ref{eq:problem-perAP-quadratic}) is defined as
\begin{align}
    &\mathcal{L}\left(\mathbf{w}_i, \boldsymbol{\Omega}_i, \mu_i\right) = \mathbf{w}_i^H \mathbf{D}_i\mathbf{w}_i - \Re\{\tilde{\mathbf{g}}_i^H\mathbf{w}_i\} + \text{tr}\left(\mathbf{T}_i\boldsymbol{\Omega}_i\right) \label{eq:Lagrangian} \\
    & \,\,\,\,\,\,+ \mu_i\left( \|\mathbf{A}_i^H\mathbf{w}_i\|^2 + \text{tr}\left( \boldsymbol{\Sigma}_i^{-1} \boldsymbol{\Omega}_i \right)- \ln \det \left(\boldsymbol{\Omega}_i\right) - \tilde{C}_{F,i}\right), \nonumber
\end{align}
with the Lagrange multiplier $\mu_i \geq 0$.

We address the dual problem (\ref{eq:dual-problem}) using the primal-dual subgradient method \cite{Nedic:JOTA09}, where the primal variables $\{\mathbf{w}_i, \boldsymbol{\Omega}_i\}$ and the dual variable $\mu_i$ are alternately optimized.
The optimal $\{\mathbf{w}_i, \boldsymbol{\Omega}_i\}$ for fixed $\mu_i$ can be obtained by making the partial derivatives of the Lagrangian with respect to each variable to zero. The resulting values are given as
\begin{subequations} \label{eq:opt-wi-Omegai}
\begin{align}
    & \mathbf{w}_i = \frac{1}{2}\left( 
\mathbf{D}_i + \mu_i \mathbf{A}_i \mathbf{A}_i^H \right)^{-1}\tilde{\mathbf{g}}_i, \label{eq:opt-wi} \\
    & \boldsymbol{\Omega}_i = \mu_i \left( 
\mathbf{T}_i + \mu_i\boldsymbol{\Sigma}_i^{-1} \right)^{-1}. \label{eq:opt-Omegai}
\end{align}
\end{subequations}
For fixed $\{\mathbf{w}_i, \boldsymbol{\Omega}_i\}$, the update of $\mu_i$ is executed in such a way that it moves toward the subgradient direction with an adaptive step size, i.e.,
\begin{align}
    \mu_i &\leftarrow \Big[ \mu_i + \delta \Big( 
\|\mathbf{A}_i^H\mathbf{w}_i\|^2 + \text{tr}\left(\boldsymbol{\Sigma}_i^{-1}\boldsymbol{\Omega}_i\right) \label{eq:update-mui-subgradient} \\
&\,\,\, \,\,\,\,\,\,- \ln \det (\boldsymbol{\Omega}_i) - \tilde{C}_{F,i} \Big) \Big]^+, \nonumber
\end{align}
with a step size $\delta > 0$ and the notation $[\cdot]^+ = \max\{\cdot, 0\}$.

\begin{algorithm}
\caption{Accelerated FP (A-FP) algorithm}

\textbf{\footnotesize{}1}~\textbf{initialize:}

\textbf{\footnotesize{}2}~Set $\{\mathbf{W},\boldsymbol{\Omega}\}$ to an arbitrary point that satisfies (\ref{eq:problem-original-psd}), and calculate the sum-rate $R_{\text{sum}}^{\text{old}}$ with the initialized $\{\mathbf{W},\boldsymbol{\Omega}\}$.

\textbf{\footnotesize{}3}~\textbf{repeat}

\textbf{\footnotesize{}4}~~~~\,Update $\{\boldsymbol{\gamma}, \boldsymbol{\theta}, \boldsymbol{\Sigma}\}$ according to (\ref{eq:optimal-gamma-theta-Sigma}).

\textbf{\footnotesize{}5}~~~~~For $i\leftarrow 1$ to $L$

\textbf{\footnotesize{}6}~~~~~~~~~~Compute $\{\mathbf{D}_i, \tilde{\mathbf{g}}_i, \mathbf{T}_i, \mathbf{A}_i\}$ with (\ref{eq:def-Di})--(\ref{eq:def-Ai}).

\textbf{\footnotesize{}7}~~~~~~~~~~Set $\mathbf{W}_i^{\text{old}}\leftarrow \mathbf{W}_i$ and $\boldsymbol{\Omega}_i^{\text{old}}\leftarrow \boldsymbol{\Omega}_i$.

\textbf{\footnotesize{}8}~~~~~~~~~~Initialize $\mu_i\leftarrow 0$, $\delta \leftarrow \delta_0$.

\textbf{\footnotesize{}9}~~~~~~~~~~\textbf{repeat}

\textbf{\footnotesize{}10}~~~~~~~~~~~~~Update $\{\mathbf{w}_i, \boldsymbol{\Omega}_i\}$ according to (\ref{eq:opt-wi-Omegai}).

\textbf{\footnotesize{}11}~~~~~~~~~~~~~Update $\mu_i$ according to (\ref{eq:update-mui-subgradient}).

\textbf{\footnotesize{}12}~~~~~~~~~~~~~Update $\delta \leftarrow \delta r$.

\textbf{\footnotesize{}13}~~~~~~~~~\textbf{until} $\|\mathbf{W}_i - \mathbf{W}_i^{\text{old}}\|_F^2 + \|\boldsymbol{\Omega}_i - \boldsymbol{\Omega}_i^{\text{old}}\|_F^2 \leq \Delta_{\text{in}}$ 

~~~~~~~~~~~~~~~~~(Otherwise, set $\mathbf{W}_i^{\text{old}}\leftarrow \mathbf{W}_i$ and $\boldsymbol{\Omega}_i^{\text{old}}\leftarrow \boldsymbol{\Omega}_i$)

\textbf{\footnotesize{}14}~~~~~~~~~If (\ref{eq:fronthaul-capacity-constraint}) is violated, update $\mathbf{W}_i \leftarrow c_i\mathbf{W}_i$, where 

~~~~~~~~~~~~$c_i\in(0,1)$ is chosen such that (\ref{eq:fronthaul-capacity-constraint}) is satisfied with

~~~~~~~~~~~~equality.

\textbf{\footnotesize{}15}~~~~Calculate the sum-rate $R_{\text{sum}}^{\text{new}}$ with the updated $\{\mathbf{W}, \boldsymbol{\Omega}\}$.

\textbf{\footnotesize{}16}~\textbf{until} $|R_{\text{sum}}^{\text{new}} - R_{\text{sum}}^{\text{old}}|\leq \Delta_{\text{out}}$ (Otherwise, set $R_{\text{sum}}^{\text{old}} \leftarrow R_{\text{sum}}^{\text{new}}$)
\end{algorithm}

In Algorithm 1, we present the detailed procedure of the proposed A-FP scheme. The decaying factor $r \in (0,1)$ is predetermined constant, and
the scaler $c_i$ used for projection in Step 14 can be found by the bisection method, given that the left-hand side of (\ref{eq:fronthaul-capacity-constraint}) with $\mathbf{W}_i \leftarrow c_i\mathbf{W}_i$ monotonically increases with $c_i$.





We briefly discuss the asymptotic complexity of optimizing the primary variables $\{\mathbf{W}, \boldsymbol{\Omega}\}$ for fixed $\{\boldsymbol{\gamma},\boldsymbol{\theta},\boldsymbol{\Sigma}\}$, which is the main computational bottleneck of the proposed FP-based scheme. When the accelerated version of FP (A-FP) method in Algorithm 1 is applied, the complexity of updating $\{\mathbf{W}, \boldsymbol{\Omega}\}$, which comprises the computations in (\ref{eq:def-Di-gtildei-Ti-Ai}), (\ref{eq:opt-wi-Omegai}) and (\ref{eq:update-mui-subgradient}), is given by $\mathcal{O}(LM^3N^3+M^2N^2K^2)$.
This complexity is significantly lower than the complexity 
$\mathcal{O}(L^4M^4N^4 + L^3 M N^3 K)$ incurred when standard convex solvers such as CVX \cite{Grant:CVX20} are used to solve the convex problem (\ref{eq:problem-FP}) for fixed $\{\boldsymbol{\gamma},\boldsymbol{\theta},\boldsymbol{\Sigma}\}$ \cite{BTal:LN19}.


\subsection{Fronthaul-Efficient Decentralized Implementation} \label{sub:decentralized}

In this subsection, we present a fronthaul-efficient decentralized implementation of the A-FP scheme in Algorithm 1. The key idea is to enable the APs and the CPU to exchange information over fronthaul links whose required data volume does not scale with the number of antennas per AP, $M$. This design allows the decentralized implementation to remain practical for CF XL-MIMO systems with $M\gg 1$.
We describe the computation tasks carried out at each AP $i$ and at the CPU, along with the information exchange required for those tasks.
We recall from Sec. \ref{sub:Uplink-Channel-Model} that each AP $i$ has access to its local CSI $\{\mathbf{h}_{i,k}\}_{k\in\mathcal{K}}$.

\subsubsection{Computation at AP $i$ and Required Information}

At each iteration, AP $i$ updates $\{\mathbf{D}_i, \tilde{\mathbf{g}}_i, \mathbf{T}_i, \mathbf{A}_i\}$ using (\ref{eq:def-Di})--(\ref{eq:def-Ai}) and then optimizes its processing variables $\{\mathbf{w}_i, \boldsymbol{\Omega}_i\}$ via iteratively computing (\ref{eq:opt-wi-Omegai}) and (\ref{eq:update-mui-subgradient}).
The auxiliary variable $\boldsymbol{\Sigma}_i$ is also updated locally, since it only depends on $\{\mathbf{w}_i, \boldsymbol{\Omega}_i\}$ and the local CSI at AP $i$, as shown in (\ref{eq:optimal-Sigma}).

For these updates, AP $i$ requires from the CPU the auxiliary variables $\boldsymbol{\gamma}$ and $\{\boldsymbol{\theta}_{i,k}\}_{k\in\mathcal{K}}$.
Additionally, updating $\tilde{\mathbf{g}}_i$ requires the information about $\{ \boldsymbol{\theta}_{i^{\prime},k}^H \mathbf{W}_{i^{\prime}}, \mathbf{h}_{i^{\prime},k} \}_{k\in\mathcal{K}, i^{\prime}\in\mathcal{L}\setminus\{i\}}$.
Since these consist of only $K$ scalars, $K$ $N$-dimensional vectors, and $K(L-1)$ scalars, respectively, the fronthaul load does not scale with the number of AP antennas $M$, unlike the centralized design in Sec. \ref{sub:FP-based}.

\subsubsection{Computation at CPU and Required Information}

At each iteration, the CPU updates $\boldsymbol{\gamma}$ and $\boldsymbol{\theta}$ using (\ref{eq:optimal-gamma}) and (\ref{eq:optimal-theta}), respectively, and sends $\boldsymbol{\gamma}$ and $\{\boldsymbol{\theta}_{i,k}\}_{k\in\mathcal{K}}$ to each AP $i$.
To perform these updates, the CPU collects from each AP $i$ the variables $\{\mathbf{W}_i \mathbf{h}_{i,k}\}_{k\in\mathcal{K}}$, $\mathbf{W}_i\mathbf{W}_i^H$ and $\boldsymbol{\Omega}_i$.
This corresponds to $K$ $N$-dimensional vectors and 2 $N\times N$ matrices, which is again independent of the number of antennas $M$.
Additionally, each AP $i$ should send $\{\boldsymbol{\theta}_{i,k}^H\mathbf{W}_i\mathbf{h}_{i,k}\}_{k\in\mathcal{K}}$, a set of $K$ scalars, to the CPU, since these values are required by the other APs for their local updates.

\begin{table}
\caption{\small Fronthaul overheads of decentralized implementation}
\vspace{-2mm}
\centering
\renewcommand{\arraystretch}{1.5}
\begin{tabular}{|c|c||c|}
    \hline    
     \multicolumn{2}{|c||} {Direction} & Fronthaul overhead \\
     \hline \hline
     \multirow{2}{*}{AP $i$ $\rightarrow$ CPU} & Variables  & \makecell[c]{$\{ \mathbf{W}_i \mathbf{h}_{i,k} \}_{k\in\mathcal{K}},\mathbf{W}_i\mathbf{W}_i^H,$ \\
     $ \boldsymbol{\Omega}_i,\{\boldsymbol{\theta}_{i,k}^H\mathbf{W}_i\mathbf{h}_{i,k} \}_{k\in\mathcal{K}}$}\\
    \cline{2-3}
       &  No. of scalars &   $2N^2+K(N+1)$ (indep. of $M$) \\
    \hline \hline
     \multirow{2}{*}{CPU $\rightarrow$ AP $i$} & Variables  &  \makecell[c]{$ \{\gamma_k\}_{k \in \mathcal{K}},\{\boldsymbol{\theta}_{i,k}\}_{k \in \mathcal{K}},$ \\
     $\{\boldsymbol{\theta}_{i^{\prime},k}^H\mathbf{W}_{i^{\prime}}\mathbf{h}_{i^{\prime},k}  \}_{k\in\mathcal{K},i^{\prime} \in \mathcal{L}\setminus \{i\}}$}\\
    \cline{2-3}
      &  No. of scalars &   $K(N+L)$ (indep. of $M$) \\
    \hline
\end{tabular}
\end{table}

In Table I, we summarize the communication overhead between each AP $i$ and the CPU over the fronthaul links in both uplink and downlink directions for the decentralized implementation described in this subsection. In the global CSI-based scheme, each AP $i$ needs to report its local CSI $\{\mathbf{h}_{i,k}\}_{k\in\mathcal{K}}$ to the CPU, resulting in a data volume that scales with the number of AP antennas $M$.
In contrast, Table I shows that in the decentralized implementation, the amount of information exchanged over the fronthaul links is independent of $M$, achieving a fronthaul-efficient and scalable design well suited for CF XL-MIMO systems.

\section{Numerical Results} \label{sec:numerical}



We consider a circular service area of a radius 100 m, where $L$ APs and $K$ UEs are uniformly distributed.
To avoid excessively strong channels, each UE is placed at least 10 m away from its nearest AP.
We assume a carrier frequency of $f_c = 3.5$ GHz and a system bandwidth $20$ MHz, following the 5G specification.
The channel vector $\mathbf{h}_{i,k}$ is given by $\mathbf{h}_{i,k} = \beta_{i,k}^{-1/2} \tilde{\mathbf{h}}_{i,k}$ \cite{Sun:VTC16}, where $\tilde{\mathbf{h}}_{i,k}\sim\mathcal{CN}(\mathbf{0}, \mathbf{I})$ represents small-scale fading effect, and $\beta_{i,k}$ is the large-scale path-loss coefficient expressed as $\beta_{i,k} = 20 \log_{10} \left(4\pi f_c / c \right) + 31 \log_{10} d_{i,k} + \chi_{\sigma}^{\text{CI}}$.
Here, $c$ denotes the speed of light, $d_{i,k}$ is the distance between AP $i$ and UE $k$, and $\chi_{\sigma}^{\text{CI}} \sim \mathcal{N}(0, (8.1)^2)$ models the log-normal shadowing effect.
Each UE transmits with a fixed power of 23 dBm, and the power spectral density of background noise at the APs is set to -169 dBm/Hz.
The initial step size $\delta_0$ and its decaying factor $r$ for Algorithm 1 are set to $\delta_0=1$ and $r=0.95$, respectively, but the algorithm was numerically verified to remain stable for $r\in [0.8, 0.98]$.

We compare the sum-rates of the proposed FP-based schemes with the following local CSI-based schemes: \textit{i)} \textit{Local EVD} \cite{Liu:Cambridge17}: The row vectors of each transformation matrix $\mathbf{W}_i$ are set to the $N$ leading eigenvectors of $\mathbf{H}_i \bar{\mathbf{P}} \mathbf{H}_i^H + \sigma_z^2\mathbf{I}$; 
\textit{ii)} \textit{Local MF} \cite{Wiffen:WCNC20}: Each transformation matrix $\mathbf{W}_i$ is set to a partial matched filter (MF) $\mathbf{W}_i = [\mathbf{h}_{i,k^i_1} \, \cdots \, \mathbf{h}_{i,k^i_N}]^H$, where $\{k_1^i, \ldots k_N^i\}$ are the indices of the UEs corresponding to the $N$ largest channel gains $p_k\|\mathbf{h}_{i,k}\|^2$ arriving at AP $i$; \textit{iii)} \textit{Local Random}: The elements of $\mathbf{W}_i$ are obtained from independent and identically distributed complex Gaussian distribution.

In the above baseline schemes, for given $\mathbf{W}_i$, the quantization noise covariance matrix $\boldsymbol{\Omega}_i$ is set to $\boldsymbol{\Omega}_i = \text{diag}(\{ \omega_{i,n} \}_{n=1}^N)$ with $\omega_{i,n} = 1/(2^{C_F/N} - 1) \mathbf{w}_{i,n}^H \mathbf{S}_i \mathbf{w}_{i,n}$.
Here $\mathbf{w}_{i,n}^H$ denotes the $n$th row vector of $\mathbf{W}_i$.
It is noted that in these schemes, the transformation matrix $\mathbf{W}_i$ and the quantization noise covariance matrix $\boldsymbol{\Omega}_i$ at each AP $i$ are designed using only the local CSI $\{\mathbf{h}_{i,k}\}_{k\in\mathcal{K}}$.
For the proposed FP-based algorithms, the initial points of $\{\mathbf{W}, \boldsymbol{\Omega}\}$ are chosen according to the local EVD scheme, unless stated otherwise.


\begin{figure}
\centering
\subfloat[vs. the number of iterations]{\includegraphics[width=0.5\linewidth, height=0.55\linewidth]{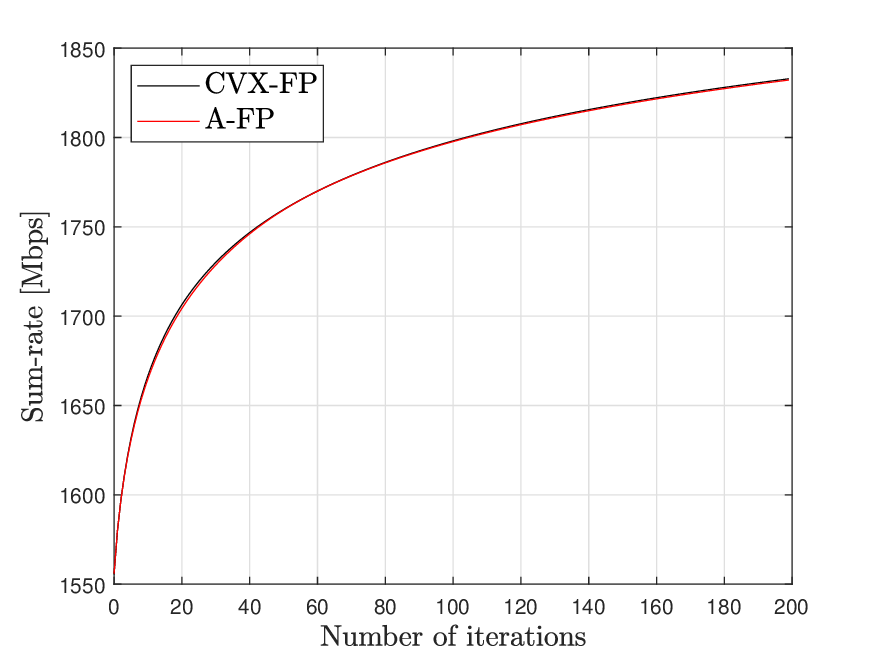}}\hfil
\subfloat[vs. the elapsed time]{\includegraphics[width=0.5\linewidth, height=0.55\linewidth]{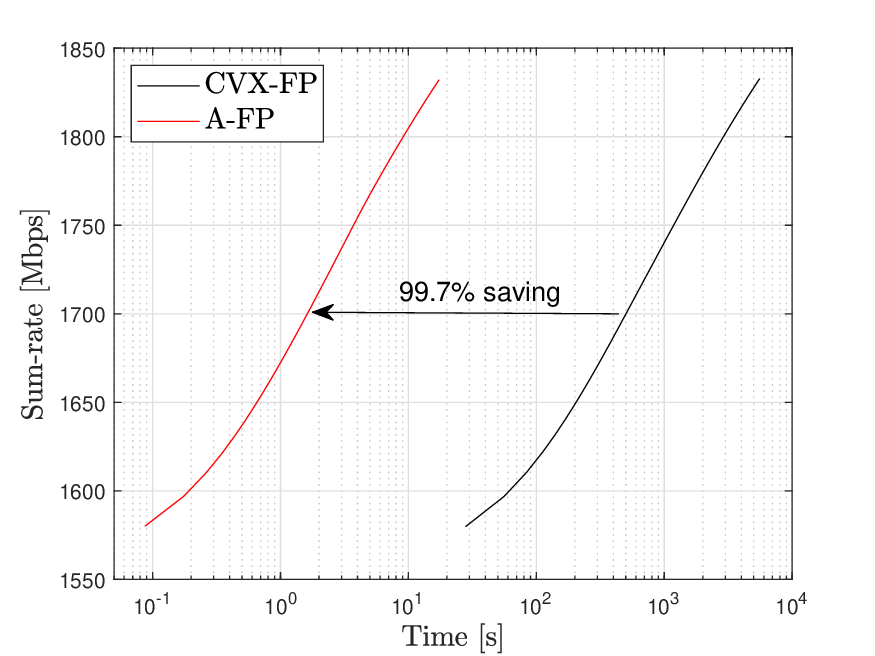}}\hfil
\caption{\small The convergence behaviors of the FP and A-FP schemes}
\label{fig:convergence}
\end{figure}

In Fig. \ref{fig:convergence}(a), we observe the convergence behaviors of the FP schemes by plotting their average sum-rates versus the number of iterations for a CF XL-MIMO system with $K=8$, $L=4$, $M=16$, and $N=2$. The CVX-FP scheme employs the CVX software \cite{Grant:CVX20} to update $\{\mathbf{W}, \boldsymbol{\Omega}\}$ at each iteration, whereas the A-FP scheme uses the per-AP closed-form updates described in Algorithm 1.
The A-FP scheme exhibits nearly the same convergence speed as the CVX-FP scheme in terms of iteration count.
Since the A-FP scheme entails a much lower per-iteration complexity than the CVX-FP scheme, as discussed in Sec. \ref{sub:distributed}, Fig. \ref{fig:convergence}(b) compares their average sum-rates with respect to the elapsed time. It is observed that the A-FP scheme can reduce the computational time by over 99$\%$ while achieving the same performance as the CVX-FP scheme. For the remaining numerical results, we focus on CF XL-MIMO systems with $M\geq 100$ and exclude the CVX-FP scheme, whose time complexity is prohibitively large while yielding almost identical performance to the A-FP scheme.

\begin{figure}
\centering\includegraphics[width=0.9\linewidth]{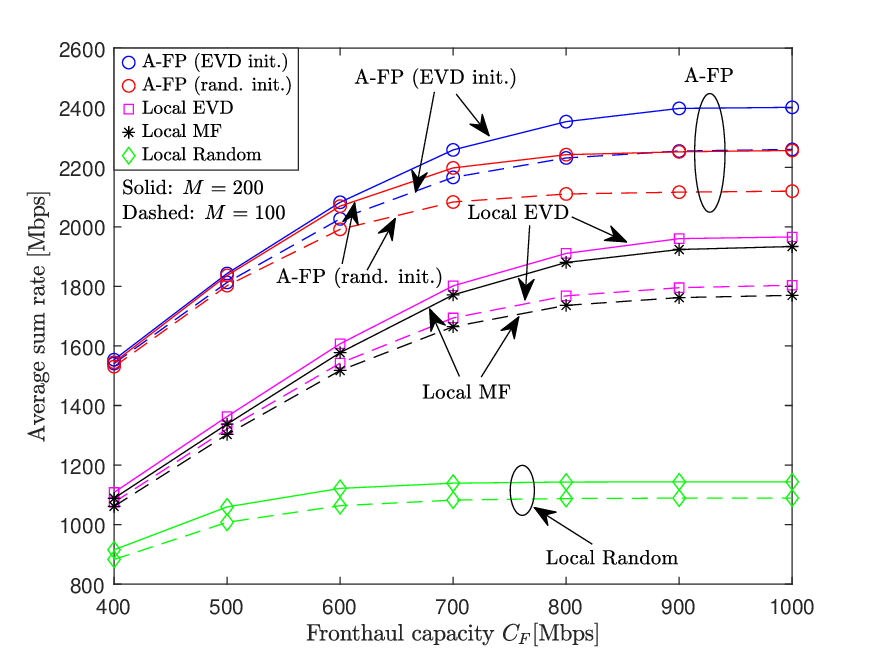}
\caption{\small Average sum-rate versus the fronthaul capacity}\label{fig:FigureB_SR_vs_CF}
\end{figure}

In Fig. \ref{fig:FigureB_SR_vs_CF}, we present the average sum-rates of various schemes as the fronthaul capacity $C_F$ increases for a CF XL-MIMO system with $K=8$, $L=4$, $M\in\{100,200\}$, and $N=2$. When $C_F$ is sufficiently large, the sum-rates of all schemes saturate at finite values due to the fixed SNR level.
Among the local CSI-based schemes, the EVD scheme, where each AP $i$ fully exploits the local covariance matrix $\mathbf{H}_i\bar{\mathbf{P}}\mathbf{H}_i^H$, achieves the highest performance.
Interestingly, the MF scheme exhibits only a slight sum-rate degradation compared to the EVD scheme, owing to the near orthogonality of the column vectors of $\mathbf{H}_i$ when $M \geq 100$.
Across all simulated cases, the proposed A-FP scheme consistently achieves substantial sum-rate gains over the local CSI-based schemes.
Specifically, the A-FP scheme achieves better performance when initialized with the local EVD solution than with a random initialization.





\section{Conclusion} \label{sec:conclusion}

We have proposed an efficient FP-based algorithm for joint optimization of linear transformation and fronthaul compression in uplink CF XL-MIMO systems. The algorithm enables closed-form per-AP updates of the optimization variables, greatly improving computational efficiency. A distributed implementation was also discussed, where the amount of information exchange between APs and CPU is independent of the number of AP antennas. Numerical results confirm that the proposed scheme significantly outperforms local CSI-based benchmarks.
As future work, we will investigate algorithms with further reduced computational complexity and designs relying solely on statistical local CSI.
Studying how to incorporate inter-AP and AP-UE synchronization effects and designing robust algorithms against imperfect synchronization are also relevant future directions \cite{Larsson:TSP24}.

\end{document}